\documentclass[preprint,12pt,fleqn]{elsarticle}

\usepackage{amsmath,amssymb}
\usepackage{amsthm}
\usepackage{enumitem}
\usepackage{lmodern}
\usepackage{microtype}
\usepackage{tikz}
\usetikzlibrary{arrows.meta,automata,positioning,calc}
\usepackage{placeins}
\usepackage{needspace}
\usepackage[hidelinks]{hyperref}
\hypersetup{pdftitle={Prime Accessible Deterministic Automata},
	pdfauthor={Zhao Shen and Chenghe Xu}}
\biboptions{numbers,sort&compress}
\journal{Theoretical Computer Science}

\allowdisplaybreaks
\setlist[enumerate]{leftmargin=2.2em,itemsep=2pt,topsep=3pt}
\setlist[itemize]{leftmargin=2.0em,itemsep=2pt,topsep=3pt}

\newtheorem{theorem}{Theorem}[section]
\newtheorem{proposition}[theorem]{Proposition}
\newtheorem{lemma}[theorem]{Lemma}
\theoremstyle{definition}

\newtheorem{example}[theorem]{Example}
\theoremstyle{remark}
\newtheorem*{remark}{Remark}

\newcommand{\A}{\mathcal{A}}
\newcommand{\B}{\mathcal{B}}
\newcommand{\C}{\mathcal{C}}
\newcommand{\U}{\mathcal{U}}
\newcommand{\FA}{F_{\A}}
\newcommand{\FB}{F_{\B}}
\newcommand{\FC}{F_{\C}}
\newcommand{\eps}{\varepsilon}
\newcommand{\N}{\mathbb N}
\newcommand{\lcmop}{\operatorname{lcm}}

\begin{document}
	\begin{frontmatter}
		
		\title{Prime Accessible Deterministic Automata}
		\author[csu]{Zhao Shen}
		\author[csu]{Chenghe Xu\corref{cor1}}
		\ead{262138051@csu.edu.cn}
		\cortext[cor1]{Corresponding author.}
		\address[csu]{School of Mathematics and Statistics, Central South University,
			Changsha, China}
		
		\begin{abstract}
			
			We classify prime accessible deterministic automata over finite alphabets,
			both among finite automata and among automata with arbitrary state sets.
			Every one-state automaton is prime. Over a one-letter alphabet, a finite
			nontrivial automaton is prime exactly when it is a cycle of prime-power
			length or a finite chain ending in a fixed state. The unique infinite
			accessible unary automaton, up to isomorphism, is also prime.
			Every nontrivial automaton over an alphabet with at least two letters is
			not prime. Shen proved the finite case; we give a shorter proof.
			
			We also characterize a decomposition of the automaton whose states are all finite input words. 
			The reachable product of the quotients by the least right congruences
			identifying the empty word with nonempty words $p$ and $q$, respectively,
			is isomorphic to this automaton exactly when
			$p^r\ne q^s$ for all positive integers $r,s$.

		\end{abstract}
		
		
	\end{frontmatter}
	
	\section{Introduction}
	
	For a nonempty set $\Sigma$ of input letters, let $\Sigma^*$ denote the set of
	all finite words over $\Sigma$, including the empty word $\eps$. For
	$w\in\Sigma^*$, write $|w|$ for its length. A deterministic automaton
	over $\Sigma$, or a $\Sigma$-automaton, is a quadruple
	$\A=(S,i,\Sigma,t)$, where $S$ is the state set, $i\in S$ is the initial state,
	and $t:S\times\Sigma\to S$ is the transition function. We extend $t$ to words by
	\[
	t(s,\eps)=s,\qquad t(s,uv)=t(t(s,u),v)
	\quad(\forall s\in S,\ \forall u,v\in\Sigma^*),
	\]
	where $uv$ denotes concatenation. The automaton is \emph{accessible} if every
	state is reached from $i$ by some word. All automata in this paper are accessible;
	no accepting states or output functions are specified.
	For a set $X$, we use $|X|$ to represent its cardinality.
	An automaton is \emph{finite} if its state set is finite, and \emph{nontrivial}
	if it has at least two states. For any automaton $\mathcal{A}$, we write
	$S_{\mathcal{A}}$, $i_{\mathcal{A}}$, and $t_{\mathcal{A}}$ for its state set,
	initial state, and transition function.
	
	We follow Yao's definitions of factors, isomorphisms, products, divisibility,
	and primality~\cite[Section~5]{yao2004}; see also Kamae and Mend\`es
	France~\cite[Section~2]{kamae1996}. Yao works with finite automata over
	alphabets with at least two letters. We use the same definitions for
	arbitrary accessible state sets and also include one-letter alphabets.
	
	For accessible $\Sigma$-automata $\A,\B$, we say that $\A$ is a
	\emph{factor} of $\B$ if there is a surjective map
	$\lambda:S_{\B}\to S_{\A}$ satisfying
	\[
	\begin{aligned}
		\lambda(i_{\B})&=i_{\A},\\
		\lambda(t_{\B}(x,a))&=t_{\A}(\lambda(x),a)
		\qquad(\forall x\in S_{\B},\ \forall a\in\Sigma).
	\end{aligned}
	\]
	We write $\lambda:\B\to\A$ for this \emph{homomorphism}. A bijective
	homomorphism is an \emph{isomorphism}; we write $\A\cong\B$ when one exists.
	
	For accessible $\Sigma$-automata $\B,\C$, their \emph{reachable synchronous
		product} is
	\[
	\B\times\C
	=\bigl(S_{\B\times\C},(i_{\B},i_{\C}),\Sigma,t_{\B\times\C}\bigr),
	\]
	where
	\[
	\begin{aligned}
		S_{\B\times\C}
		&=\{(t_{\B}(i_{\B},w),t_{\C}(i_{\C},w)):w\in\Sigma^*\},\\
		t_{\B\times\C}((x,y),a)
		&=(t_{\B}(x,a),t_{\C}(y,a))
	\end{aligned}
	\]
	for every $(x,y)\in S_{\B\times\C}$ and $a\in\Sigma$. Only states reachable
	from the initial pair are retained. Throughout the paper, ``product'' means
	this reachable synchronous product.
	
	We say that $\B$ is \emph{divisible by} $\A$, or that $\B$ is a
	\emph{multiple} of $\A$, and write $\A\mid\B$, if there is an accessible
	$\Sigma$-automaton $\C$ such that
	\[
	\B\cong\A\times\C.
	\]
	Proposition~\ref{prop:factor} shows that divisibility is equivalent to the
	factor relation even when state sets are infinite.
	
	Over a fixed alphabet, we consider two classes: all accessible deterministic
	automata, and those with finite state sets. An automaton $\A$ in either
	class is \emph{prime in that class} if
	\[
	\A\mid\B\times\C
	\quad\Longrightarrow\quad
	\bigl(\A\mid\B\ \text{or}\ \A\mid\C\bigr)
	\]
	for all $\B,\C$ in that class. We call $\A$ \emph{finite-state prime} when
	$\B,\C$ are required to be finite, and \emph{unrestrictedly prime} when
	their state sets may be arbitrary. Outside Section~\ref{sec:finite},
	``prime'' means unrestrictedly prime unless stated otherwise. One-state
	automata are included.
	
	This differs from the primality studied by Kupferman and
	Mosheiff~\cite{km2015} and Spenner~\cite{spenner2026}, where a regular language
	is composite if it is a finite intersection of languages accepted by
	deterministic finite automata with fewer states than its minimal automaton.
	Here primality concerns the underlying accessible transition structure,
	with no size restriction on the two factors.
	
	Shen proved that no finite nontrivial automaton over a finite alphabet with
	at least two letters is prime among finite automata~\cite{shen2022}.
	Theorem~\ref{thm:nonfin} gives a shorter proof. It uses the states reached
	after the first letter instead of classifying the whole transition graph.
	The required state information is stored in two finite automata, with
	explicit bounds on their sizes.
	
	The main new results are Lemma~\ref{lem:prefix} and
	Theorem~\ref{thm:nonunary}. They determine exactly when two quotients by the least right congruences
	identifying the empty word with chosen nonempty words have product
	isomorphic to the automaton of all finite input words.
	This gives more than non-primality: the product retains the entire input
	word, not just the state reached in a particular automaton.
	The resulting decomposition applies to arbitrary state sets and alphabets.
	
	An automaton with one input letter is called \emph{unary}; one with at least
	two input letters is called \emph{non-unary}. Over the alphabet $\{a\}$,
	for integers $\ell\ge0$ and $p\ge1$, define the unary automaton
	\[
	\A(\ell,p)=\bigl(\{0,1,\ldots,\ell+p-1\},0,\{a\},\delta_{\ell,p}\bigr),
	\]
	where
	\[
	\delta_{\ell,p}(r,a)=
	\begin{cases}
		r+1,&0\le r<\ell+p-1,\\
		\ell,&r=\ell+p-1.
	\end{cases}
	\]
	Starting from $0$, this automaton makes $\ell$ transitions before entering a
	cycle of $p$ states. We call $\ell$ the \emph{tail length} and $p$ the
	\emph{cycle length}.
	
	The following is the primality classification of finite-state automata.
	
	\begin{theorem}\label{thm:finite}
		Let $\A=(S,i,\Sigma,t)$ be a finite accessible deterministic automaton over a finite nonempty
		alphabet. Primality in this theorem is taken among finite accessible deterministic automata.
		\begin{enumerate}[label=\textup{(\roman*)}]
			\item If $|S|=1$, then $\A$ is prime.
			\item If $|\Sigma|=1$, then $\A$ is isomorphic to a unique automaton $\A(\ell,p)$, with
			integers $\ell\ge0$ and $p\ge1$.
			Apart from the one-state automaton $\A(0,1)$, it is prime if and only if
			\begin{equation*}
				\ell=0\ \text{and $p$ is a prime power},
				\qquad\text{or}\qquad
				\ell\ge1\ \text{and }p=1.
			\end{equation*}
			Here a prime power means an integer of the form $q^k$, where $q$ is a prime integer and
			$k\ge1$ is an integer.
			\item If $|\Sigma|\ge2$ and $|S|\ge2$, then there are finite accessible
			$\Sigma$-automata $\B,\C$ such that
			\[
			\A\mid\B\times\C,\qquad \A\nmid\B,\qquad \A\nmid\C.
			\]
			In particular, $\A$ is not prime.
		\end{enumerate}
	\end{theorem}
	
	We now allow both finite and infinite state sets. Put $\N=\{0,1,2,\ldots\}$.
	The \emph{infinite unary successor automaton} is
	\[
	\U=(\N,0,\{a\},\delta),\qquad
	\delta(n,a)=n+1\quad(\forall n\in\N).
	\]
	Thus each input letter increases the state by one. An automaton with one
	state is called a \emph{unit automaton}.
	
	\begin{theorem}\label{thm:main}
		Let $\A=(S,i,\Sigma,t)$ be an accessible deterministic automaton over a finite nonempty
		alphabet. Primality in this theorem is taken among all accessible deterministic automata,
		with no restriction on state-set cardinality.
		\begin{enumerate}[label=\textup{(\roman*)}]
			\item Every unit automaton is prime.
			\item If $|\Sigma|=1$ and $S$ is infinite, then $\A$ is isomorphic to the infinite unary
			successor automaton $\U$ and is prime.
			\item If $|\Sigma|=1$ and $S$ is finite, the classification is the same as in Theorem~\ref{thm:finite}\textup{(ii)}.
			\item If $|\Sigma|\ge2$ and $|S|\ge2$, then $\A$ is not prime.
		\end{enumerate}
	\end{theorem}
	
	\section{A divisibility criterion}\label{sec:prelim}
	
	The equivalence below is stated by Yao for finite automata~\cite[Proposition~5 and pp.~259--260]{yao2004}. We include the proof because it applies without any finiteness assumption.
	
	\begin{proposition}\label{prop:factor}
		For accessible automata $\A$ and $\B$ over the same alphabet,
		\begin{equation*}
			\A\mid\B
			\quad\Longleftrightarrow\quad
			\A\text{ is a factor of }\B.
		\end{equation*}
	\end{proposition}
	
	\begin{proof}
		If $\B\cong\A\times\C$, compose this isomorphism with the projection
		onto $\A$. The projection is a surjective homomorphism because $\A$
		is accessible. Thus $\A$ is a factor of $\B$.
		
		Conversely, let $\lambda:\B\to\A$ be a surjective homomorphism.
		Since
		\[
		\lambda(t_{\B}(i_{\B},w))=t_{\A}(i_{\A},w)
		\qquad(\forall w\in\Sigma^*),
		\]
		the states of $\A\times\B$ are exactly
		$(\lambda(x),x)$ with $x\in S_{\B}$. Hence
		$x\mapsto(\lambda(x),x)$ is a bijection from $S_{\B}$ onto this
		state set. It preserves the initial state and commutes with every
		transition because $\lambda$ does. This gives
		$\B\cong\A\times\B$, so $\A\mid\B$.
	\end{proof}
	
	If $\A$ and $\B$ are finite, Proposition~\ref{prop:factor}
	shows that $\A\mid\B$ has the same meaning in both classes: its proof
	takes the auxiliary factor to be $\B$ itself. Only the range of
	$\B,\C$ in the definition of primality changes.
	
	For $\A=(S,i,\Sigma,t)$, define its \emph{reaching map} $\FA:\Sigma^*\to S$ by
	\begin{equation*}
		\FA(w)=t(i,w)
		\qquad(\forall w\in\Sigma^*).
	\end{equation*}
	
	An automaton \emph{identifies} two words if they reach the same state, and
	\emph{distinguishes} them if they reach different states.
	
	\begin{lemma}\label{lem:div}
		For accessible automata $\A$ and $\B$ over $\Sigma$,
		\[
		\A\mid\B\quad\Longleftrightarrow\quad
		\left[
		\begin{gathered}
			\forall u,v\in\Sigma^*,\\
			\FB(u)=\FB(v)\ \Longrightarrow\ \FA(u)=\FA(v)
		\end{gathered}
		\right].
		\]
	\end{lemma}
	
	\begin{proof}
		If $\A\mid\B$, Proposition~\ref{prop:factor} gives a surjective
		homomorphism $\lambda:\B\to\A$. Induction on word length gives
		$\lambda(\FB(w))=\FA(w)$ for every $w\in\Sigma^*$.
		Thus $\FB(u)=\FB(v)$ implies $\FA(u)=\FA(v)$.
		
		Conversely, assume the implication in the statement. For each
		$x\in S_{\B}$, accessibility of $\B$ gives a word $w$ with
		$\FB(w)=x$. Define
		\[
		\lambda(x)=\FA(w).
		\]
		If another word $v$ also reaches $x$ in $\B$, then
		$\FB(v)=\FB(w)$, so the assumed implication gives
		$\FA(v)=\FA(w)$. Thus $\lambda$ is well defined on all of
		$S_{\B}$. It is surjective: for every $s\in S_{\A}$,
		accessibility of $\A$ gives a word $u$ with $\FA(u)=s$, and
		then $\lambda(\FB(u))=s$.
		
		Taking $w=\eps$ gives $\lambda(i_{\B})=i_{\A}$.
		For any $x\in S_{\B}$ and $a\in\Sigma$, choose $w$ with
		$\FB(w)=x$. Then
		\[
		\lambda(t_{\B}(x,a))=\lambda(\FB(wa))
		=\FA(wa)=t_{\A}(\FA(w),a)=t_{\A}(\lambda(x),a).
		\]
		Hence $\lambda$ is a surjective homomorphism, and
		Proposition~\ref{prop:factor} gives $\A\mid\B$.
	\end{proof}
	
	\begin{remark}
		By Lemma~\ref{lem:div}, $\A\mid\B\times\C$ if and only if
		\[
		\left.
		\begin{aligned}
			\FB(u)&=\FB(v)\\
			\FC(u)&=\FC(v)
		\end{aligned}
		\right\}
		\ \Longrightarrow\ \FA(u)=\FA(v)
		\qquad(\forall u,v\in\Sigma^*).
		\]
		To prove $\A\nmid\B$, it suffices to find two words identified by
		$\B$ but distinguished by $\A$.
	\end{remark}
	
	\section{Finite-state automata}\label{sec:finite}
	
	Throughout this section the alphabet is finite. Every automaton used to
	prove non-primality also has a finite state set. Thus \emph{prime} means
	finite-state prime.
	
	\subsection{Unit automata}
	
	\begin{proposition}\label{prop:unit}
		Every unit automaton is unrestrictedly prime, and hence also finite-state prime.
	\end{proposition}
	
	\begin{proof}
		Let $\A=(\{i\},i,\Sigma,t_{\A})$. For every accessible $\Sigma$-automaton $\B$, the constant map $S_{\B}\to\{i\}$ is surjective, preserves the initial state, and commutes with every transition. Thus $\A$ is a factor of $\B$, and Proposition~\ref{prop:factor} gives $\A\mid\B$. In particular, whenever $\A\mid\B\times\C$, one already has $\A\mid\B$; hence $\A$ is prime.
	\end{proof}
	
	\subsection{Finite unary automata}\label{sec:unary}
	
	Throughout this subsection, let $\Sigma=\{a\}$.
	For a word $w\in\Sigma^*$ and $n\in\N$, write $w^n$ for the
	concatenation of $n$ copies of $w$, with $w^0=\eps$. We use this
	notation over any alphabet.
	
	Every finite accessible unary automaton is isomorphic to a unique
	$\A(\ell,p)$: the first repeated state determines the tail length $\ell$
	and cycle length $p$. For all $m,n\in\N$ with $m<n$,
	\begin{equation}\label{eq:unary}
		\delta_{\ell,p}(0,a^m)=\delta_{\ell,p}(0,a^n)
		\quad\Longleftrightarrow\quad
		\bigl(m\ge\ell\ \text{and}\ p\mid(n-m)\bigr).
	\end{equation}
	
	\begin{lemma}\label{lem:unary}
		For integers $\ell_1,\ell_2\ge0$ and $p_1,p_2\ge1$,
		\begin{equation*}
			\A(\ell_1,p_1)\mid\A(\ell_2,p_2)
			\quad\Longleftrightarrow\quad
			\bigl(\ell_1\le\ell_2\ \text{and}\ p_1\mid p_2\bigr).
		\end{equation*}
	\end{lemma}
	
	\begin{proof}
		By Lemma~\ref{lem:div}, divisibility requires that whenever two words
		reach the same state in $\A(\ell_2,p_2)$, they also reach the
		same state in $\A(\ell_1,p_1)$.
		Applying this to $a^{\ell_2}$ and $a^{\ell_2+p_2}$ gives
		$\ell_1\le\ell_2$ and $p_1\mid p_2$ by \eqref{eq:unary}.
		Conversely, these two conditions and \eqref{eq:unary} preserve
		every such equality, so Lemma~\ref{lem:div} applies.
	\end{proof}
	
	For positive integers $r,s$, let $\lcmop(r,s)$ denote their
	least common multiple.
	
	\begin{lemma}\label{lem:product}
		For integers $\ell_1,\ell_2\ge0$ and $p_1,p_2\ge1$,
		\begin{equation*}
			\A(\ell_1,p_1)\times\A(\ell_2,p_2)
			\cong\A\bigl(\max\{\ell_1,\ell_2\},\lcmop(p_1,p_2)\bigr).
		\end{equation*}
	\end{lemma}
	
	\begin{proof}
		The words $a^m,a^n$, where $m,n\in\N$ and $m<n$, reach the same
		product state exactly when they reach the same state in both
		coordinates. By \eqref{eq:unary}, this means
		\[
		m\ge\max\{\ell_1,\ell_2\},\qquad
		\lcmop(p_1,p_2)\mid(n-m).
		\]
		The unique form $\A(\ell,p)$ of the product therefore has the
		stated parameters.
	\end{proof}
	
	\begin{theorem}\label{thm:unary}
		Let $\ell\ge0$ and $p\ge1$ be integers with $\ell+p\ge2$. The automaton $\A(\ell,p)$ is finite-state prime if and only if
		\begin{equation*}
			\ell=0\ \text{and $p$ is a prime power},
			\qquad\text{or}\qquad
			\ell\ge1\ \text{and }p=1.
		\end{equation*}
	\end{theorem}
	
	\begin{proof}
		Let $\B,\C$ be finite accessible unary automata such that
		$\A(\ell,p)\mid\B\times\C$. Up to isomorphism, write
		$\B=\A(\ell_1,p_1)$ and
		$\C=\A(\ell_2,p_2)$, with integers $\ell_1,\ell_2\ge0$
		and $p_1,p_2\ge1$. The preceding lemmas give
		\[
		\ell\le\max\{\ell_1,\ell_2\},\qquad
		p\mid\lcmop(p_1,p_2).
		\]
		If $p=1$, the tail inequality holds for one factor.
		If $\ell=0$ and $p$ is a prime power, then $p$ divides $p_1$ or
		$p_2$: the exponent of a prime in a least common multiple is the
		larger of its exponents in the two integers. Lemma~\ref{lem:unary}
		proves primality in both cases.
		
		For $\ell\ge1$ and $p\ge2$, we have
		\[
		\A(\ell,p)\cong\A(\ell,1)\times\A(0,p).
		\]
		If $\ell=0$ and $p>1$ is not a prime power, write $p=rs$ with
		coprime integers $r,s>1$. Then
		\[
		\A(0,p)\cong\A(0,r)\times\A(0,s).
		\]
		In each case, Lemma~\ref{lem:unary} shows that $\A(\ell,p)$
		divides neither displayed factor.
	\end{proof}
	
	\subsection{Non-unary alphabets: a finite proof}\label{sec:nonunary}
	
	For $u\in\Sigma^*$, write
	$u\Sigma^*=\{uv:v\in\Sigma^*\}$ for the set of words beginning with $u$.
	
	The non-primality conclusion below is Shen's theorem~\cite{shen2022}.
	We give a shorter proof and record the state bounds supplied by the
	construction. For $\A=(S,i,\Sigma,t)$, put
	\begin{equation*}
		H_c=\{\FA(cw):w\in\Sigma^*\}
		\qquad(\forall c\in\Sigma).
	\end{equation*}
	
	\begin{theorem}\label{thm:nonfin}
		Let $\A=(S,i,\Sigma,t)$ be finite and accessible over a finite alphabet,
		with $n=|S|\ge2$ and $|\Sigma|\ge2$. There are finite accessible
		$\Sigma$-automata $\B,\C$ with
		\[
		|S_{\B}|\le n+2,\qquad |S_{\C}|\le n^2+n+2,
		\]
		such that
		\begin{equation*}
			\A\mid\B\times\C,
			\qquad
			\A\nmid\B,
			\qquad
			\A\nmid\C.
		\end{equation*}
		Consequently, $\A$ is not finite-state prime.
	\end{theorem}
	
	\begin{proof}
		We distinguish two cases.
		
		\smallskip
		\noindent\emph{Case 1: some $H_a$ contains at least two states, where $a\in\Sigma$.}
		Choose $b\in\Sigma\setminus\{a\}$. There is a word $u\in a\Sigma^*$ with $1\le |u|\le2$ such that
		\begin{equation*}
			\FA(u)\ne\FA(b).
		\end{equation*}
		Indeed, take $u=a$ if $\FA(a)\ne\FA(b)$. Otherwise $\FA(a)=\FA(b)$; since $H_a$ contains at least two states, some $d\in\Sigma$ satisfies $t(\FA(a),d)\ne\FA(a)$, and then $u=ad$ works.
		
		To construct $\B$, use two distinct new state labels $r_B,f$ and a distinct label $\bar s$ for each $s\in S$, all different from $r_B,f$. The bar is only a label, not an operation on $S$. Define
		\[
		S_{\B}=\{r_B,f\}\cup
		\{\overline{\FA(cv)}:c\in\Sigma\setminus\{a\},\ v\in\Sigma^*\},
		\qquad
		\B=(S_{\B},r_B,\Sigma,t_{\B}),
		\]
		with transitions, for $c\in\Sigma$,
		\[
		t_{\B}(r_B,c)=
		\begin{cases}
			f,&c=a,\\
			\overline{\FA(c)},&c\ne a,
		\end{cases}
		\qquad
		t_{\B}(f,c)=f,
		\]
		and
		\[
		t_{\B}(\bar s,c)=\overline{t(s,c)}
		\qquad(\forall s\in S\text{ with }\bar s\in S_{\B},\ \forall c\in\Sigma).
		\]
		These transitions stay in $S_{\B}$: if $\bar s=\overline{\FA(dv)}$ for $d\in\Sigma\setminus\{a\}$ and $v\in\Sigma^*$, then $\overline{t(s,c)}=\overline{\FA(dvc)}\in S_{\B}$ for every $c\in\Sigma$. The transition depends only on $s$, not on the word used to reach it. Induction on word length gives, for every $w\in\Sigma^*$,
		\begin{equation*}
			\FB(w)=
			\begin{cases}
				r_B,&w=\eps,\\
				f,&w\in a\Sigma^*,\\
				\overline{\FA(w)},&w\ne\eps\text{ and }w\notin a\Sigma^*.
			\end{cases}
		\end{equation*}
		Thus every state in $S_{\B}$ is reachable: $r_B$ is reached by $\eps$, $f$ by $a$, and each remaining state by a word occurring in its definition. Also $|S_{\B}|\le |S|+2$. Since $H_a$ has at least two states, there exist $v_1,v_2\in\Sigma^*$ with $\FA(av_1)\ne\FA(av_2)$, whereas $\FB(av_1)=\FB(av_2)=f$. Lemma~\ref{lem:div} therefore gives
		\begin{equation*}
			\A\nmid\B.
		\end{equation*}
		
		To construct $\C$, let $P$ consist of the proper prefixes of $u$:
		\[
		P=\{x\in\Sigma^*:u=xv\text{ for some }v\in\Sigma^*\setminus\{\eps\}\}.
		\]
		In particular $\eps\in P$ and $|P|=|u|\le2$. Use state labels $r_x$ for $x\in P$, $p_s$ for $s\in S$, and $q_{s_1,s_2}$ for $s_1,s_2\in S$. Labels from different families, or with different indices within one family, are distinct. Define
		\[
		\begin{aligned}
			S_{\C}={}&\{r_x:x\in P\}\\
			&\cup\{p_{\FA(w)}:
			w\in\Sigma^*\setminus(P\cup u\Sigma^*\cup b\Sigma^*)\}\\
			&\cup\{q_{\FA(uz),\FA(bz)}:z\in\Sigma^*\},
		\end{aligned}
		\]
		and put $\C=(S_{\C},r_{\eps},\Sigma,t_{\C})$. For $x\in P$ and $c\in\Sigma$, set
		\[
		t_{\C}(r_x,c)=
		\begin{cases}
			r_{xc},&xc\in P,\\
			q_{\FA(u),\FA(b)},&xc=u\text{ or }xc=b,\\
			p_{\FA(xc)},&\text{otherwise}.
		\end{cases}
		\]
		For each state of the indicated form in $S_{\C}$ and every $c\in\Sigma$, set
		\[
		t_{\C}(p_s,c)=p_{t(s,c)},
		\qquad
		t_{\C}(q_{s_1,s_2},c)=q_{t(s_1,c),\,t(s_2,c)}.
		\]
		The first transition rule is unambiguous: $u$ starts with $a\ne b$, and neither $u$ nor $b$ belongs to $P$. For $x\in P$, if $xc$ starts with $u$, then $xc=u$ because $|xc|\le|u|$; if $xc$ starts with $b$, then $x=\eps$ and $xc=b$. Hence the remaining clause produces a state in $S_{\C}$.
		
		The other transitions also stay in $S_{\C}$. If a word $w\in\Sigma^*$ does not belong to $P$ and begins with neither $u$ nor $b$, then the same is true of $wc$ for every $c\in\Sigma$. Otherwise $w$ would already be a proper prefix of $u$ or begin with $u$ or $b$. Thus every transition from a state $p_{\FA(w)}$ has its prescribed target in $S_{\C}$. A transition from $q_{\FA(uz),\FA(bz)}$, where $z\in\Sigma^*$, has target $q_{\FA(uzc),\FA(bzc)}$, also in $S_{\C}$. The rules depend only on the displayed state indices, so they are independent of the words representing those indices.
		
		Induction on word length now gives
		\begin{equation}\label{eq:pair}
			\begin{aligned}
				\FC(x)&=r_x &&(\forall x\in P),\\
				\FC(uz)=\FC(bz)&=q_{\FA(uz),\FA(bz)}
				&&(\forall z\in\Sigma^*),\\
				\FC(w)&=p_{\FA(w)}
				&&(\forall w\in\Sigma^*\setminus(P\cup u\Sigma^*\cup b\Sigma^*)).
			\end{aligned}
		\end{equation}
		These formulas show that every state in $S_{\C}$ is reachable. Moreover,
		\[
		|S_{\C}|\le |P|+|S|+|S|^2\le |S|^2+|S|+2,
		\]
		so $\C$ is finite. Since $\FC(u)=\FC(b)$ but $\FA(u)\ne\FA(b)$, Lemma~\ref{lem:div} gives
		\begin{equation*}
			\A\nmid\C.
		\end{equation*}
		
		Finally, let $x,y\in\Sigma^*$ satisfy $\FB(x)=\FB(y)$ and $\FC(x)=\FC(y)$. If their common state in $\B$ is $r_B$, then $x=y=\eps$. If it is $\bar s$ for some $s\in S$, then $\FA(x)=\FA(y)=s$. Otherwise both words begin with $a$. In this last case, a common state $r_z$ in $\C$, with $z\in P$, gives $x=y=z$, while a common state $p_s$, with $s\in S$, gives $\FA(x)=\FA(y)=s$. If their common state is $q_{s_1,s_2}$, with $s_1,s_2\in S$, both words must begin with $u$, because the alternative in \eqref{eq:pair} begins with $b\ne a$. The first state index then gives $\FA(x)=\FA(y)=s_1$. We have proved
		\[
		\left.
		\begin{aligned}
			\FB(x)&=\FB(y)\\
			\FC(x)&=\FC(y)
		\end{aligned}
		\right\}
		\ \Longrightarrow\ \FA(x)=\FA(y)
		\qquad(\forall x,y\in\Sigma^*).
		\]
		By the definition of the product and Lemma~\ref{lem:div}, this proves $\A\mid\B\times\C$.
		
		\smallskip
		\noindent\emph{Case 2: every $H_c$ contains exactly one state.}
		Put
		\begin{equation*}
			s_c=\FA(c),
			\qquad
			Q=\{s_c:c\in\Sigma\},
			\qquad
			k=|Q|.
		\end{equation*}
		This assumption gives $t(s_c,d)=\FA(cd)=s_c$ for all $c,d\in\Sigma$, so every $s_c$ is \emph{absorbing}, meaning that every input letter fixes it. Accessibility gives $S=\{i\}\cup Q$. Moreover $i\notin Q$: if $\FA(c)=i$ for some $c$, then
		\begin{equation*}
			\FA(w)=t(i,w)=t(\FA(c),w)=\FA(cw)=i
		\end{equation*}
		for every $w\in\Sigma^*$, contrary to $|S|\ge2$. Hence $S=\{i\}\sqcup Q$, where $\sqcup$ denotes disjoint union.
		
		Suppose first that $k\ge2$. Write $\mathbb Z/k\mathbb Z$ for the set of integer
		residue classes modulo $k$, with addition taken modulo $k$. Choose a bijection $\gamma:Q\to\mathbb Z/k\mathbb Z$. Let $\B$ have states $\{r_B\}\sqcup\mathbb Z/k\mathbb Z$, initial state $r_B$, and the following transitions. Here $x\xrightarrow{d}y$ means that the letter $d$ sends state $x$ to state $y$.
		\begin{equation*}
			r_B\mathbin{\xrightarrow{d}}0,
			\qquad
			z\mathbin{\xrightarrow{d}}z+1
			\quad(\forall d\in\Sigma,\ \forall z\in\mathbb Z/k\mathbb Z).
		\end{equation*}
		Let $\C$ have states $\{r_C\}\sqcup\mathbb Z/k\mathbb Z$ and initial state $r_C$, with transitions for $d\in\Sigma$ and $z\in\mathbb Z/k\mathbb Z$ given by
		\begin{equation*}
			r_C\mathbin{\xrightarrow{d}}\gamma(s_d),
			\qquad
			z\mathbin{\xrightarrow{d}}z+1.
		\end{equation*}
		Both automata are accessible: word lengths $1,\ldots,k$ reach all
		residues of $\B$, while a suitable one-letter word reaches each
		residue of $\C$. Each has $k+1=n$ states. If $w=dv$ with $d\in\Sigma$ and $v\in\Sigma^*$, then
		\begin{equation*}
			\FB(w)=|v|,
			\qquad
			\FC(w)=\gamma(s_d)+|v|
			\quad\pmod{k}.
		\end{equation*}
		Applying $\gamma^{-1}$ to $\FC(w)-\FB(w)$ recovers $s_d=\FA(w)$, and the pair of initial states occurs only at $\eps$. Therefore $\A\mid\B\times\C$ by Lemma~\ref{lem:div}.
		
		Choose letters $a,b$ with $s_a\ne s_b$. Then $\FB(a)=\FB(b)$, so $\A\nmid\B$. Let $1\le r<k$ represent $\gamma(s_b)-\gamma(s_a)$ modulo $k$, and choose any $c\in\Sigma$. Since $s_a$ is absorbing,
		\begin{equation*}
			\FC(ac^r)=\gamma(s_a)+r=\gamma(s_b)=\FC(b),
			\qquad
			\FA(ac^r)=s_a\ne s_b=\FA(b).
		\end{equation*}
		Thus $\A\nmid\C$.
		
		Finally suppose $k=1$. Then $S=\{i,s\}$ and every nonempty word reaches the absorbing state $s$. Choose distinct letters $a,b$. Let $\B$ have states $0,1$, with initial state $0$, the letter $a$ fixing $0$, and every other letter sending $0$ to the absorbing state $1$. Define $\C$ symmetrically using $b$. Both are accessible and have two states, so the stated bounds hold. The automaton $\B$ identifies $\eps$ with $a$, and $\C$ identifies $\eps$ with $b$, so $\A$ divides neither. Write $a^*=\{a^n:n\in\N\}$ and $b^*=\{b^n:n\in\N\}$.
		For every $w\in\Sigma^*$,
		\begin{equation*}
			\FB(w)=0\Longleftrightarrow w\in a^*,
			\qquad
			\FC(w)=0\Longleftrightarrow w\in b^*.
		\end{equation*}
		Since $a^*\cap b^*=\{\eps\}$, the product distinguishes the empty word from every nonempty word, which is exactly the distinction made by $\A$. Hence $\A\mid\B\times\C$.
	\end{proof}
	
	\subsection{Completion of the finite-state classification}
	
	\begin{proof}[Proof of Theorem~\ref{thm:finite}]
		Part~\textup{(i)} is Proposition~\ref{prop:unit}. For part~\textup{(ii)}, the tail--cycle description gives the unique form $\A(\ell,p)$, and Theorem~\ref{thm:unary} gives exactly the stated arithmetic alternatives. Part~\textup{(iii)} is Theorem~\ref{thm:nonfin}; the automata $\B,\C$ constructed there are finite.
	\end{proof}
	
	\begin{example}\label{ex:tm}
		Let $\Sigma=\{0,1\}$ and consider the Thue--Morse automaton
		$\mathcal T=(\{0,1\},0,\Sigma,\tau)$, where
		\[
		\tau(s,0)=s,\qquad \tau(s,1)=1-s
		\qquad(\forall s\in\{0,1\}).
		\]
		For a word $w\in\Sigma^*$, let $|w|_1$ denote the number of
		occurrences of $1$ in $w$. Then
		\[
		F_{\mathcal T}(w)=|w|_1\bmod 2
		\qquad(\forall w\in\Sigma^*).
		\]
		Thus the state reached by the binary expansion of $n\in\N$ is
		the $n$th term of the Thue--Morse sequence, starting with index $0$.
		
		Case~1 of the proof of Theorem~\ref{thm:nonfin} applies with
		$a=0$, $b=1$, and $u=0$. It gives automata $\B,\C$ with initial
		states $r_B,r_C$, shown with $\mathcal T$ in Figure~\ref{fig:tm}.
		Unattached arrows mark initial states; edge labels are input letters.
		In the notation of that proof, $q_0$ represents the pair $(0,1)$
		and $q_1$ represents $(1,0)$. Both automata are accessible:
		$\eps,0,11,1$ reach $r_B,f,\bar0,\bar1$ in $\B$, and
		$\eps,0,01$ reach $r_C,q_0,q_1$ in $\C$, respectively.
		
		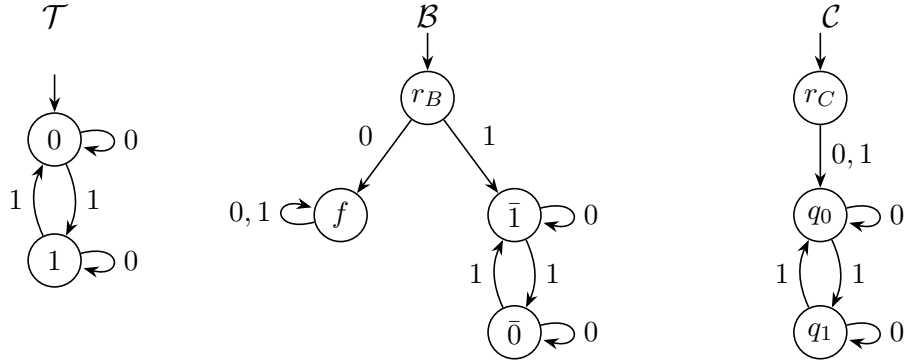
\begin{figure}[htbp]
			\centering
			\begin{minipage}[t]{.22\linewidth}
				\centering\vspace{0pt}
				$\mathcal T$\par\smallskip
				\begin{tikzpicture}[
					>=Stealth,->,semithick,initial text=,initial distance=5mm,
					every node/.style={font=\small},
					every state/.style={draw,circle,minimum size=7mm,inner sep=1pt}]
					\path[use as bounding box] (-1.1,-3.65) rectangle (1.1,.8);
					\node[state,initial,initial where=above] (t0) at (0,-.55) {$0$};
					\node[state] (t1) at (0,-2.15) {$1$};
					\path
					(t0) edge[loop right] node {$0$} (t0)
					edge[bend left=25] node[right] {$1$} (t1)
					(t1) edge[loop right] node {$0$} (t1)
					edge[bend left=25] node[left] {$1$} (t0);
				\end{tikzpicture}
			\end{minipage}\hfill
			\begin{minipage}[t]{.44\linewidth}
				\centering\vspace{0pt}
				$\B$\par\smallskip
				\begin{tikzpicture}[
					>=Stealth,->,semithick,initial text=,initial distance=5mm,
					every node/.style={font=\small},
					every state/.style={draw,circle,minimum size=7mm,inner sep=1pt}]
					\path[use as bounding box] (-2.3,-3.65) rectangle (2.3,.8);
					\node[state,initial,initial where=above] (br) at (0,0) {$r_B$};
					\node[state] (bf) at (-1.15,-1.55) {$f$};
					\node[state] (b1) at (1.15,-1.55) {$\bar1$};
					\node[state] (b0) at (1.15,-3.1) {$\bar0$};
					\path
					(br) edge node[above left] {$0$} (bf)
					edge node[above right] {$1$} (b1)
					(bf) edge[loop left] node {$0,1$} (bf)
					(b1) edge[loop right] node {$0$} (b1)
					edge[bend left=25] node[right] {$1$} (b0)
					(b0) edge[loop right] node {$0$} (b0)
					edge[bend left=25] node[left] {$1$} (b1);
				\end{tikzpicture}
			\end{minipage}\hfill
			\begin{minipage}[t]{.28\linewidth}
				\centering\vspace{0pt}
				$\C$\par\smallskip
				\begin{tikzpicture}[
					>=Stealth,->,semithick,initial text=,initial distance=5mm,
					every node/.style={font=\small},
					every state/.style={draw,circle,minimum size=7mm,inner sep=1pt}]
					\path[use as bounding box] (-1.1,-3.65) rectangle (1.4,.8);
					\node[state,initial,initial where=above] (cr) at (0,0) {$r_C$};
					\node[state] (c0) at (0,-1.55) {$q_0$};
					\node[state] (c1) at (0,-3.1) {$q_1$};
					\path
					(cr) edge node[right] {$0,1$} (c0)
					(c0) edge[loop right] node {$0$} (c0)
					edge[bend left=25] node[right] {$1$} (c1)
					(c1) edge[loop right] node {$0$} (c1)
					edge[bend left=25] node[left] {$1$} (c0);
				\end{tikzpicture}
			\end{minipage}
			\caption{The automata $\mathcal T,\B,\C$ in Example~\ref{ex:tm}.}
			\label{fig:tm}
		\end{figure}
		
		Figure~\ref{fig:tmprod} shows their product. It has exactly the
		five states in the following table.
		Each is reached by the word in the middle column, and the set is
		closed under both letter transitions. The last column defines
		a map $\lambda:S_{\B\times\C}\to\{0,1\}$.
		\[
		\begin{array}{c|c|c}
			\text{product state}&\text{reaching word}&\lambda\\\hline
			(r_B,r_C)&\eps&0\\
			(f,q_0)&0&0\\
			(f,q_1)&01&1\\
			(\bar1,q_0)&1&1\\
			(\bar0,q_1)&11&0
		\end{array}
		\]
		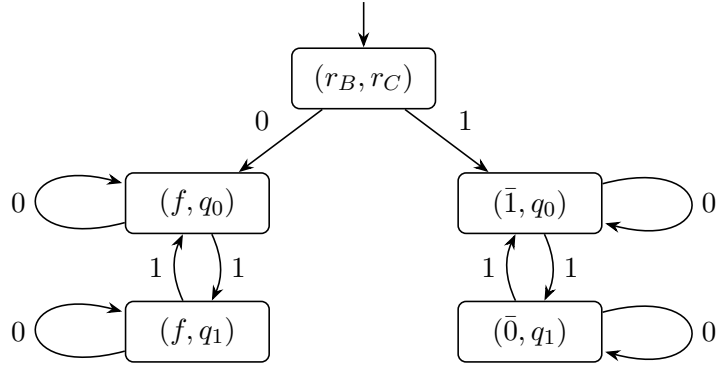
\begin{figure}[htbp]
			\centering
			\begin{tikzpicture}[
				>=Stealth,->,semithick,initial text=,initial distance=6mm,
				every node/.style={font=\small},
				every state/.style={draw,rectangle,rounded corners=3pt,
					minimum width=19mm,minimum height=8mm,inner sep=4pt}]
				\node[state,initial,initial where=above] (pr) at (0,0) {$(r_B,r_C)$};
				\node[state] (pf0) at (-2.2,-1.65) {$(f,q_0)$};
				\node[state] (pf1) at (-2.2,-3.35) {$(f,q_1)$};
				\node[state] (pb1) at (2.2,-1.65) {$(\bar1,q_0)$};
				\node[state] (pb0) at (2.2,-3.35) {$(\bar0,q_1)$};
				\path
				(pr) edge node[above left] {$0$} (pf0)
				edge node[above right] {$1$} (pb1)
				(pf0) edge[loop left] node {$0$} (pf0)
				edge[bend left=25] node[right] {$1$} (pf1)
				(pf1) edge[loop left] node {$0$} (pf1)
				edge[bend left=25] node[left] {$1$} (pf0)
				(pb1) edge[loop right] node {$0$} (pb1)
				edge[bend left=25] node[right] {$1$} (pb0)
				(pb0) edge[loop right] node {$0$} (pb0)
				edge[bend left=25] node[left] {$1$} (pb1);
			\end{tikzpicture}
			\caption{The reachable product $\B\times\C$ in Example~\ref{ex:tm}.}
			\label{fig:tmprod}
		\end{figure}
		
		The diagrams show that $\lambda$ preserves the initial
		state and commutes with both letters. It is surjective, so
		$\mathcal T\mid\B\times\C$. However, $\B$ identifies $0,01$,
		whereas $\mathcal T$ distinguishes them. Likewise, $\C$
		identifies $0,1$, whereas $\mathcal T$ distinguishes them.
		Lemma~\ref{lem:div} gives $\mathcal T\nmid\B$ and
		$\mathcal T\nmid\C$. Hence the Thue--Morse automaton is not
		finite-state prime.
	\end{example}
	\FloatBarrier
	
	\section{Allowing arbitrary state sets}\label{sec:all}
	
	In this section prime means unrestrictedly prime. The proof of
	Theorem~\ref{thm:nonfin} already gives non-primality for arbitrary state
	sets over a finite alphabet: the construction in Case~1 does not require
	finiteness, while Case~2 gives
	$S=\{i\}\cup\{\FA(c):c\in\Sigma\}$, which is finite.
	After treating unary automata, we prove a stronger decomposition result
	in Lemma~\ref{lem:prefix} and Theorem~\ref{thm:nonunary}:
	the product of two explicit quotients recovers every input word.
	This second construction also permits infinite alphabets.
	
	\subsection{Unary automata}
	
	We again work over the alphabet $\{a\}$.
	
	\begin{proposition}\label{prop:inf}
		Every accessible infinite unary automaton is isomorphic to $\U$.
		Every accessible unary automaton divides $\U$.
	\end{proposition}
	
	\begin{proof}
		For $\A=(S,i,\{a\},t)$, the map $n\mapsto t(i,a^n)$ from $\N$
		onto $S$ is a surjective homomorphism, so $\A\mid\U$.
		If two values with indices $m<n$ were equal, every later state
		would lie among the first $n$ states. Thus, when $S$ is infinite,
		this map is also injective and is an isomorphism.
	\end{proof}
	
	\begin{proposition}\label{prop:unary}
		For a finite accessible unary automaton, finite-state primality and unrestricted primality 
		are equivalent. Hence the arithmetic classification in Theorem~\ref{thm:unary} is unchanged 
		when $\B,\C$ are allowed to be infinite.
	\end{proposition}
	
	\begin{proof}
		Unrestricted primality implies finite-state primality. For the
		converse, let $\A$ be finite-state prime and
		suppose that $\A\mid\B\times\C$ for accessible unary automata
		$\B,\C$. If both factors are finite, $\A$ divides one of them
		by assumption. Otherwise, Proposition~\ref{prop:inf} shows that
		an infinite factor is isomorphic to $\U$ and is divisible by $\A$.
	\end{proof}
	
	\begin{theorem}\label{thm:inf}
		The infinite unary successor automaton $\U$ is unrestrictedly prime.
	\end{theorem}
	
	\begin{proof}
		If $\U\mid\B\times\C$, then the product is infinite, so at least
		one factor is infinite. By Proposition~\ref{prop:inf}, that factor
		is isomorphic to $\U$. Hence $\U$ divides one factor.
	\end{proof}
	
	\subsection{A decomposition by prefix reduction}\label{sec:prefix}
	
	We work over any nonempty alphabet $\Sigma$.
	The following right congruences give the decomposition used in
	Theorem~\ref{thm:nonunary}.
	
	An equivalence relation $\sim$ on $\Sigma^*$ is a \emph{right congruence} if
	\[
	u\sim v\quad\Longrightarrow\quad uw\sim vw
	\qquad(\forall u,v,w\in\Sigma^*).
	\]
	
	For a nonempty word $p\in\Sigma^*$, write each $w\in\Sigma^*$ uniquely as
	\[
	w=p^m\nu_p(w),\qquad \nu_p(w)\notin p\Sigma^*,
	\]
	where $m\in\N$ is maximal. This is possible because
	$m\le |w|/|p|$. The word $\nu_p(w)$ is obtained by deleting all initial
	copies of $p$. In particular, $\nu_p(\eps)=\nu_p(p)=\eps$.
	Deleting copies from the beginning does not increase word length.
	
	\Needspace{12\baselineskip}
	\begin{lemma}\label{lem:prefix}
		For each nonempty $p\in\Sigma^*$, the relation $\sim_p$ defined by
		\[
		u\sim_p v\quad\Longleftrightarrow\quad\nu_p(u)=\nu_p(v)
		\qquad(\forall u,v\in\Sigma^*)
		\]
		is the smallest right congruence identifying $\eps$ and $p$.
		For any nonempty $p,q\in\Sigma^*$,
		\[
		{\sim_p}\cap{\sim_q}=\Delta_{\Sigma^*}
		\quad\Longleftrightarrow\quad
		p^r\ne q^s\quad\text{for all integers }r,s\ge1,
		\]
		where $\Delta_{\Sigma^*}=\{(w,w):w\in\Sigma^*\}$ is the equality
		relation. In particular, the intersection is the equality relation
		whenever $p,q$ have different first letters.
	\end{lemma}
	
	\begin{proof}
		Equality of the values of $\nu_p$ defines an equivalence relation.
		If $u\sim_p v$, write $u=p^mz$ and $v=p^nz$, where
		$z=\nu_p(u)=\nu_p(v)$ and $m,n\in\N$. For every $x\in\Sigma^*$,
		deleting the displayed copies of $p$ and then any further initial
		copies gives
		\[
		\nu_p(ux)=\nu_p(zx)=\nu_p(vx).
		\]
		Thus $\sim_p$ is a right congruence and identifies $\eps,p$.
		Any right congruence identifying these words identifies $z$ with
		$p^mz$ for every $z\in\Sigma^*$ and $m\in\N$.
		It therefore contains $\sim_p$.
		
		Suppose that distinct words $u,v$ are equivalent under both
		$\sim_p$ and $\sim_q$. Two distinct words with the same
		$\nu_p$-value have different lengths. After interchanging $u,v$,
		assume $|u|>|v|$. Their common reduced representative gives
		$u=p^rv$ for an integer $r\ge1$. The same argument for $q$ gives
		$u=q^sv$ for an integer $s\ge1$. Cancelling the common suffix
		$v$ from $p^rv=q^sv$ yields $p^r=q^s$.
		Conversely, if $p^r=q^s=w$ for some $r,s\ge1$, then
		$\eps\ne w$, but $\eps\sim_p w$ and $\eps\sim_q w$.
		This proves the equivalence. Words with different first letters
		cannot have equal positive powers.
	\end{proof}
	
	If $\sim$ is an equivalence relation on a set $X$, write
	$[x]_{\sim}=\{y\in X:x\sim y\}$ for the equivalence class of $x\in X$,
	and $X/{\sim}=\{[x]_{\sim}:x\in X\}$ for the set of its classes.
	
	For each nonempty $p\in\Sigma^*$, define the quotient automaton
	\[
	\mathcal E_p
	=\bigl(\Sigma^*/{\sim_p},[\eps]_{\sim_p},\Sigma,e_p\bigr),
	\]
	where
	\[
	e_p([w]_{\sim_p},c)=[wc]_{\sim_p}
	\qquad(\forall w\in\Sigma^*,\ \forall c\in\Sigma).
	\]
	Lemma~\ref{lem:prefix} makes the transition independent of the chosen
	representative. Every class $[w]_{\sim_p}$ is reached by $w$, so
	$\mathcal E_p$ is accessible. Moreover,
	\[
	[w]_{\sim_p}
	=\{p^m\nu_p(w):m\in\N\}\qquad(\forall w\in\Sigma^*).
	\]
	Thus $\nu_p(w)$ is the unique shortest word in its class. We may
	identify the states with the words in $\Sigma^*\setminus p\Sigma^*$.
	In this description, the initial state is $\eps$ and the transition is
	\[
	z\longmapsto\nu_p(zc)
	\qquad(\forall z\in\Sigma^*\setminus p\Sigma^*,\ \forall c\in\Sigma).
	\]
	
	The \emph{free word automaton} is
	\[
	\mathcal W=(\Sigma^*,\eps,\Sigma,d),\qquad
	d(w,c)=wc\quad(\forall w\in\Sigma^*,\ \forall c\in\Sigma).
	\]
	Its state reached by $w$ is $w$ itself. The next theorem gives an exact
	criterion for the two quotients to recover this automaton.
	
	\begin{theorem}\label{thm:nonunary}
		Let $p,q$ be nonempty words over a nonempty alphabet $\Sigma$. Then
		\begin{equation}\label{eq:wordprod}
			\mathcal E_p\times\mathcal E_q\cong\mathcal W
			\quad\Longleftrightarrow\quad
			p^r\ne q^s\quad\text{for all integers }r,s\ge1.
		\end{equation}
		If $|\Sigma|\ge2$ and $\A$ is a nontrivial accessible
		$\Sigma$-automaton, there are words $p,q$ with different first
		letters, $|p|=1$ and $1\le|q|\le2$, such that
		\[
		\mathcal E_p\times\mathcal E_q\cong\mathcal W,\qquad
		\A\mid\mathcal E_p\times\mathcal E_q,
		\qquad \A\nmid\mathcal E_p,\quad\A\nmid\mathcal E_q.
		\]
		In particular, $\A$ is not prime. No finiteness assumption on
		$S_{\A}$ or $\Sigma$ is needed.
	\end{theorem}
	
	\begin{proof}
		The map
		\[
		\phi(w)=([w]_{\sim_p},[w]_{\sim_q})
		\qquad(\forall w\in\Sigma^*)
		\]
		is a surjective homomorphism from $\mathcal W$ to
		$\mathcal E_p\times\mathcal E_q$. It is the unique homomorphism
		between these automata, since it must send each word to the state
		reached by that word. Hence an isomorphism exists exactly when
		$\phi$ is injective. Lemma~\ref{lem:prefix} gives the criterion
		in \eqref{eq:wordprod}.
		
		Let $\A=(S,i,\Sigma,t)$ be nontrivial and accessible, with
		$|\Sigma|\ge2$. Some letter $a$ satisfies $\FA(a)\ne i$;
		otherwise every word would fix $i$. Choose $b\ne a$ and put
		\[
		p=a,\qquad
		q=\begin{cases}
			b,&\FA(b)\ne i,\\
			ba,&\FA(b)=i.
		\end{cases}
		\]
		In the second case $\FA(ba)=t(\FA(b),a)=t(i,a)\ne i$.
		Thus $p,q$ have different first letters and satisfy
		$\FA(p)\ne i$ and $\FA(q)\ne i$.
		The automaton $\mathcal E_p$ identifies $\eps,p$, whereas $\A$
		distinguishes them; $\eps,q$ give the same obstruction for
		$\mathcal E_q$. Lemma~\ref{lem:div} proves both non-divisibilities.
		
		By \eqref{eq:wordprod}, every reachable product state has the form
		$([w]_{\sim_p},[w]_{\sim_q})$ for a unique $w\in\Sigma^*$.
		The map
		\[
		\lambda([w]_{\sim_p},[w]_{\sim_q})=\FA(w)
		\qquad(\forall w\in\Sigma^*)
		\]
		is therefore well defined. It is surjective by accessibility,
		preserves the initial state, and commutes with each letter because
		$\FA(wc)=t(\FA(w),c)$. This proves
		$\A\mid\mathcal E_p\times\mathcal E_q$.
	\end{proof}
	
	When $p,q$ have different first letters, at least one of
	$\nu_p(w),\nu_q(w)$ equals $w$. The inverse of $\phi$ then chooses
	the longer coordinate; equal-length coordinates are both $w$.
	This recovery rule uses the different-first-letter assumption, not
	merely the absence of a common positive power.
	If $|\Sigma|\ge2$, each $\mathcal E_p$ is infinite: for a letter $c$
	different from the first letter of $p$, the words $c^m$ ($m\in\N$)
	are distinct reduced representatives. Thus the finite constructions
	of Section~\ref{sec:finite} are still needed when the factors must
	be finite.
	
	\Needspace{10\baselineskip}
	\begin{example}\label{ex:binary}
		Let $\Sigma=\{0,1\}$ and consider the infinite-state automaton
		\[
		\mathcal D=(\N,0,\Sigma,t_{\mathcal D}),\qquad
		t_{\mathcal D}(n,c)=2n+c
		\quad(\forall n\in\N,\ \forall c\in\Sigma).
		\]
		Here each input letter is also regarded as the integer $0$ or $1$.
		For $w\in\Sigma^*$, let $\operatorname{val}_2(w)$ be its binary
		value, defined by
		\[
		\operatorname{val}_2(\eps)=0,\qquad
		\operatorname{val}_2(wc)=2\operatorname{val}_2(w)+c
		\quad(\forall w\in\Sigma^*,\ \forall c\in\Sigma).
		\]
		Thus words are read from left to right, and
		$F_{\mathcal D}(w)=\operatorname{val}_2(w)$ for every
		$w\in\Sigma^*$. Every natural number is reached by its binary
		expansion, so $\mathcal D$ is accessible. The first states and
		transitions are shown in Figure~\ref{fig:binary}(a).
		
		Since $F_{\mathcal D}(1)=1$ and $F_{\mathcal D}(0)=0$, the proof
		of Theorem~\ref{thm:nonunary} chooses $p=1$ and $q=01$.
		The states of $\mathcal E_1$ may be represented by
		$\{\eps\}\cup0\Sigma^*$, and those of $\mathcal E_{01}$ by
		$\Sigma^*\setminus01\Sigma^*$. Their initial states are $\eps$,
		and their transitions append the input letter and then delete
		all initial copies of $1$ or $01$, respectively.
		Figures~\ref{fig:binary}(b) and~\ref{fig:prefix}(a) show these
		automata.
		
		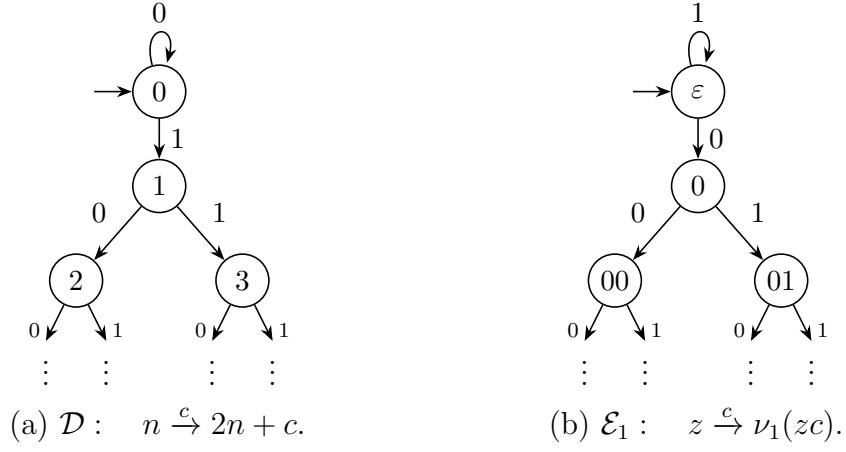
\begin{figure}[htbp]
			\centering
			\begin{minipage}[t]{.48\textwidth}
				\centering
				\begin{tikzpicture}[
					>=Stealth,->,semithick,font=\small,
					every state/.style={draw,circle,minimum size=7mm,inner sep=1pt},
					initial text=,initial distance=5mm]
					\node[state,initial] (d0) at (0,0) {$0$};
					\node[state] (d1) at (0,-1.25) {$1$};
					\node[state] (d2) at (-1.1,-2.5) {$2$};
					\node[state] (d3) at (1.1,-2.5) {$3$};
					\path
					(d0) edge[loop above] node {$0$} (d0)
					edge node[right] {$1$} (d1)
					(d1) edge node[above left] {$0$} (d2)
					edge node[above right] {$1$} (d3);
					\foreach \q in {d2,d3}{
						\draw (\q) -- ++(-.4,-.8)
						node[pos=.7,left,font=\scriptsize] {$0$}
						node[below,inner sep=1pt] {$\vdots$};
						\draw (\q) -- ++(.4,-.8)
						node[pos=.7,right,font=\scriptsize] {$1$}
						node[below,inner sep=1pt] {$\vdots$};
					}
				\end{tikzpicture}
				
				\smallskip
				\textup{(a)} \(\mathcal D:\quad n\xrightarrow{c}2n+c.\)
			\end{minipage}\hfill
			\begin{minipage}[t]{.48\textwidth}
				\centering
				\begin{tikzpicture}[
					>=Stealth,->,semithick,font=\small,
					every state/.style={draw,circle,minimum size=7mm,inner sep=1pt},
					initial text=,initial distance=5mm]
					\node[state,initial] (e0) at (0,0) {$\eps$};
					\node[state] (e1) at (0,-1.25) {$0$};
					\node[state] (e2) at (-1.1,-2.5) {$00$};
					\node[state] (e3) at (1.1,-2.5) {$01$};
					\path
					(e0) edge[loop above] node {$1$} (e0)
					edge node[right] {$0$} (e1)
					(e1) edge node[above left] {$0$} (e2)
					edge node[above right] {$1$} (e3);
					\foreach \q in {e2,e3}{
						\draw (\q) -- ++(-.4,-.8)
						node[pos=.7,left,font=\scriptsize] {$0$}
						node[below,inner sep=1pt] {$\vdots$};
						\draw (\q) -- ++(.4,-.8)
						node[pos=.7,right,font=\scriptsize] {$1$}
						node[below,inner sep=1pt] {$\vdots$};
					}
				\end{tikzpicture}
				
				\smallskip
				\textup{(b)} \(\mathcal E_1:\quad z\xrightarrow{c}\nu_1(zc).\)
			\end{minipage}
			\caption{The automata $\mathcal D$ and $\mathcal E_1$ in Example~\ref{ex:binary}.}
			\label{fig:binary}
		\end{figure}
		
		In $\mathcal E_1$, the words $\eps,1$ reach the same state;
		in $\mathcal E_{01}$, the words $\eps,01$ reach the same state.
		Since
		\[
		F_{\mathcal D}(\eps)=0,\qquad
		F_{\mathcal D}(1)=F_{\mathcal D}(01)=1,
		\]
		Lemma~\ref{lem:div} shows that $\mathcal D$ divides neither factor.
		
		Figure~\ref{fig:prefix}(b) labels the product state
		$(\nu_1(w),\nu_{01}(w))$ by the word $w$; each transition appends
		its input letter. Since $1$ and $01$ have different first letters,
		the longer coordinate recovers $w$. Thus
		\[
		\lambda\bigl(\nu_1(w),\nu_{01}(w)\bigr)
		=\operatorname{val}_2(w)
		\qquad(\forall w\in\Sigma^*)
		\]
		is well defined. It is surjective, preserves the initial state,
		and commutes with both letters by the recursion for
		$\operatorname{val}_2$. Hence
		$\mathcal D\mid\mathcal E_1\times\mathcal E_{01}$, although
		$\mathcal D$ divides neither factor.
		
		\begin{figure}[htbp]
			\centering
			\begin{minipage}[t]{.48\textwidth}
				\centering
				\begin{tikzpicture}[
					>=Stealth,->,semithick,font=\small,
					every state/.style={draw,circle,minimum size=7mm,inner sep=1pt},
					initial text=,initial distance=5mm]
					\node[state,initial above] (e0) at (0,0) {$\eps$};
					\node[state] (e1) at (-1.25,-1.25) {$0$};
					\node[state] (e2) at (1.25,-1.25) {$1$};
					\node[state] (e3) at (-1.7,-2.5) {$00$};
					\node[state] (e4) at (.5,-2.5) {$10$};
					\node[state] (e5) at (2,-2.5) {$11$};
					\path
					(e0) edge node[above left] {$0$} (e1)
					edge node[above right] {$1$} (e2)
					(e1) edge[bend left=55] node[above left] {$1$} (e0)
					edge node[left] {$0$} (e3)
					(e2) edge node[above left] {$0$} (e4)
					edge node[above right] {$1$} (e5);
					\foreach \q in {e3,e4,e5}{
						\draw (\q) -- ++(-.35,-.8)
						node[pos=.7,left,font=\scriptsize] {$0$}
						node[below,inner sep=1pt] {$\vdots$};
						\draw (\q) -- ++(.35,-.8)
						node[pos=.7,right,font=\scriptsize] {$1$}
						node[below,inner sep=1pt] {$\vdots$};
					}
				\end{tikzpicture}
				
				\smallskip
				\textup{(a)} \(\mathcal E_{01}:\quad z\xrightarrow{c}\nu_{01}(zc).\)
			\end{minipage}\hfill
			\begin{minipage}[t]{.48\textwidth}
				\centering
				\begin{tikzpicture}[
					>=Stealth,->,semithick,font=\small,
					every state/.style={draw,circle,minimum size=7mm,inner sep=1pt},
					initial text=,initial distance=5mm]
					\node[state,initial above] (w0) at (0,0) {$\eps$};
					\node[state] (w1) at (-1.3,-1.25) {$0$};
					\node[state] (w2) at (1.3,-1.25) {$1$};
					\node[state] (w3) at (-1.95,-2.5) {$00$};
					\node[state] (w4) at (-.65,-2.5) {$01$};
					\node[state] (w5) at (.65,-2.5) {$10$};
					\node[state] (w6) at (1.95,-2.5) {$11$};
					\path
					(w0) edge node[above left] {$0$} (w1)
					edge node[above right] {$1$} (w2)
					(w1) edge node[above left] {$0$} (w3)
					edge node[above right] {$1$} (w4)
					(w2) edge node[above left] {$0$} (w5)
					edge node[above right] {$1$} (w6);
					\foreach \q in {w3,w4,w5,w6}{
						\draw (\q) -- ++(-.3,-.8)
						node[pos=.7,left,font=\scriptsize] {$0$}
						node[below,inner sep=1pt] {$\vdots$};
						\draw (\q) -- ++(.3,-.8)
						node[pos=.7,right,font=\scriptsize] {$1$}
						node[below,inner sep=1pt] {$\vdots$};
					}
				\end{tikzpicture}
				
				\smallskip
				\textup{(b)} \(\mathcal E_1\times\mathcal E_{01}\cong\mathcal W.\)
			\end{minipage}
			\caption{The automaton $\mathcal E_{01}$ and the product $\mathcal E_1\times\mathcal E_{01}$.}
			\label{fig:prefix}
		\end{figure}
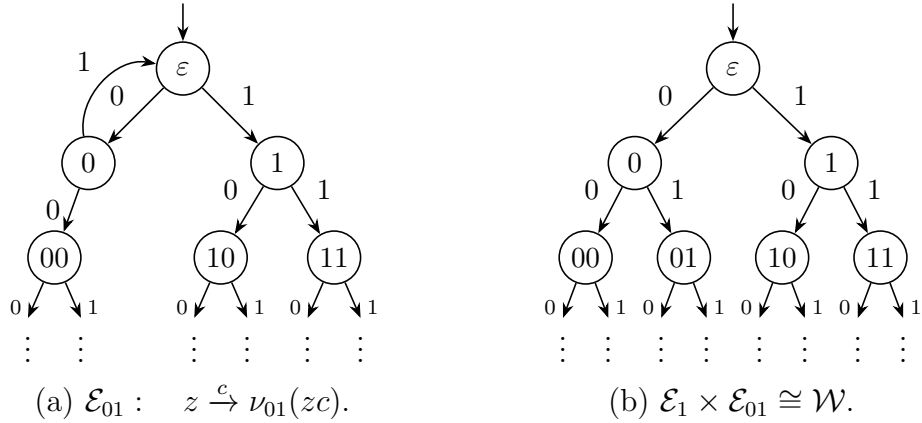
	\end{example}
	\FloatBarrier
	
	\begin{remark}
		The dots in Figures~\ref{fig:binary} and~\ref{fig:prefix} indicate
		continuing branches. In $\mathcal E_1$, every state $z\ne\eps$ has
		transitions $z\xrightarrow{c}zc$ for $c\in\{0,1\}$. The same rule
		holds in $\mathcal E_{01}$ for $z\notin\{\eps,0\}$.
	\end{remark}
	
	\subsection{Completion of the classification}\label{sec:complete}
	
	\begin{proof}[Proof of Theorem~\ref{thm:main}]
		Proposition~\ref{prop:unit} proves part~\textup{(i)}.
		If the alphabet has one letter, Proposition~\ref{prop:inf} shows
		that an infinite automaton is isomorphic to $\U$, which is prime
		by Theorem~\ref{thm:inf}. This proves part~\textup{(ii)}.
		For a finite unary automaton, Theorem~\ref{thm:unary} gives the
		classification among finite automata, and Proposition~\ref{prop:unary}
		shows that it is unchanged when infinite factors are allowed.
		This proves part~\textup{(iii)}. Theorem~\ref{thm:nonunary} proves
		part~\textup{(iv)}.
	\end{proof}
	
	\begin{remark}
		Theorems~\ref{thm:finite} and~\ref{thm:main} show that finite-state
		and unrestricted primality coincide for finite automata over finite
		alphabets. By Theorem~\ref{thm:nonunary}, Theorem~\ref{thm:main} also
		holds for arbitrary alphabets; over an infinite alphabet, only unit
		automata are prime.
	\end{remark}
	\section*{Acknowledgments}
	Z. Shen was supported by the National Natural Science Foundation of
	China (Grant No.~12301011).
	
	\section*{Declaration of generative AI and AI-assisted technologies
		in the manuscript preparation process}
	During the writing of the paper, the authors used AI tools (ChatGPT and DeepSeek)
	to refine some proof details, adjust the structure of the paper, and improve the exposition. 
	All mathematical statements and proofs in the paper were verified by the authors. The authors take full
	responsibility for the paper.
	

\end{document}